\pdfoutput=1

\documentclass[11pt,a4paper]{article}
\usepackage{times,latexsym}
\usepackage{url}
\usepackage[T1]{fontenc}

\usepackage[acceptedWithA]{tacl2021v1}

\usepackage{xspace,mfirstuc,tabulary}

\newif\iftaclinstructions
\taclinstructionsfalse 
\iftaclinstructions
\renewcommand{\confidential}{}
\renewcommand{\anonsubtext}{(No author info supplied here, for consistency with
TACL-submission anonymization requirements)}
\newcommand{\instr}
\fi

\iftaclpubformat 

\else

\fi

\usepackage{amsmath}
\usepackage{amssymb}
\usepackage{amsthm}
\usepackage{xcolor}
\usepackage{paralist}
\usepackage{thm-restate}

\usepackage{mathtools}
\usepackage{multirow}
\usepackage{bbm}

\DeclarePairedDelimiter\floor{\lfloor}{\rfloor}

\DeclarePairedDelimiter\abs{\lvert}{\rvert}%
\DeclarePairedDelimiter\norm{\lVert}{\rVert}%

\usepackage{stmaryrd}

\makeatletter
\let\oldabs\abs
\def\abs{\@ifstar{\oldabs}{\oldabs*}}
\let\oldnorm\norm
\def\norm{\@ifstar{\oldnorm}{\oldnorm*}}
\makeatother

\newtheorem{theorem}{Theorem}
\newtheorem{lemma}{Lemma}

\newtheorem{corollary}{Corollary}[theorem]

\theoremstyle{definition}
\newtheorem{definition}{Definition}

\def\Snospace~{\S{}}

\newcommand{\AC}{\mathsf{AC}}
\newcommand{\NC}{\mathsf{NC}}
\newcommand{\TC}{\mathsf{TC}}
\newcommand{\C}{\mathsf{C}}

\newcommand{\poly}{\mathrm{poly}}

\renewcommand{\O}{\mathrm{O}}

\newcommand{\dplus}{d_\oplus}

\usepackage{listings}

\lstdefinelanguage{RASP}
{
  morekeywords={select, selector_width, aggregate}
}

\definecolor{OliveGreen}{RGB}{0,102,51}
\renewcommand{\O}{\mathrm{O}}

\newcommand{\new}[1]{{\color{blue} #1}}
\newcommand{\newB}[1]{{\color{blue} #1}}

\renewcommand{\new}[1]{#1}
\renewcommand{\newB}[1]{#1}

\newcommand{\circuit}{\mathrm{Circuit}}
\newcommand{\gate}{\mathrm{Gate}}
\newcommand{\dir}{\mathrm{Dir}}
\newcommand{\argument}{\mathrm{Arg}}
\newcommand{\xx}{\texttt{X}}
\newcommand{\lequal}{\texttt{<=}}
\newcommand{\gequal}{\texttt{>=}}

\usepackage[linesnumbered,ruled,noend]{algorithm2e}

\newcommand{\sizeof}{\mathrm{sizeof}}

\renewcommand{\rshift}{\mathrm{rshift}}

\usepackage{graphicx}

\title{The Parallelism Tradeoff: Limitations of Log-Precision Transformers}

\author{William Merrill \\
  Center for Data Science \\
  New York University, New York, NY \\
  \texttt{willm@nyu.edu} \\\And
  Ashish Sabharwal \\
  Allen Institute for AI \\
  Seattle, WA \\
  \texttt{ashishs@allenai.org} \\}

\begin{document}

\maketitle

\begin{abstract}
Despite their omnipresence in modern NLP, characterizing the computational power of transformer neural nets remains an interesting open question. We prove that transformers whose \new{arithmetic precision is logarithmic} in the number of input tokens (and whose feedforward nets are computable using space linear in their input) can be simulated by constant-depth logspace-uniform threshold circuits. This provides insight on the power of transformers using known results in complexity theory. For example, if $\mathsf L \neq \mathsf P$ (i.e., not all poly-time problems can be solved using logarithmic space), then transformers cannot even \new{accurately} solve linear equalities or check membership in an arbitrary context-free grammar \new{with empty productions}. Our result intuitively emerges from the transformer architecture's high parallelizability. We thus speculatively introduce the idea of a fundamental \textbf{parallelism tradeoff}: any model architecture as parallelizable as the transformer will obey limitations similar to it. Since parallelism is key to training models at massive scale, this suggests a potential inherent weakness of the scaling paradigm.
\end{abstract}

\section{Introduction}

This work aims to characterize the computational model implicit in transformer neural networks \citep{vaswani2017attention},
which form the basis of recent breakthroughs in large language models such as BERT~\cite{devlin-etal-2019-bert}, T5~\cite{t5}, and GPT-3~\cite{gpt3}.
What computational primitives can the transformer's components implement, and what problems can the full system solve in aggregate? These questions are important for interpreting transformers in a principled way, understanding potential limitations of their reasoning capabilities, and building trust in deployed transformer-based systems.

Early theoretical work on transformers established their Turing completeness, albeit with assumptions like infinite precision and arbitrarily powerful feedforward subnets~\citep{perez2019on,Dehghani2019UniversalT}.
On the other hand, a strand of more recent work uses techniques from circuit complexity theory to derive strong limitations on the types of problems transformers can solve given restrictions on the form of attention
\new{allowed in the transformer}.
Specifically, \citet{hahn-2020-theoretical} and \citet{hao2022} showed transformers \new{restricted to} hard attention \new{are very limited: they} can only solve problems in a weak complexity class (non-uniform $\AC^0$) that doesn't even contain basic problems like majority of $n$ bits.
\citet{merrill2022SatAttnTC0} extended this to a more general class of ``saturated attention'' transformers with a floating point datatype, and showed a larger class of problems (non-uniform $\TC^0$) as an upper bound.
This motivates analyzing a setting that strikes a middle ground: \emph{Can we characterize transformers whose precision and feedforward nets' computational power are \underline{realistically bounded}, but where attention is also \underline{realistically expressive}?}

\new{An important practical limitation of these prior results is the ``non-uniform'' nature of the considered circuit classes, which makes these classes non-realizable and the findings difficult to interpret. This is because non-uniform $\AC^0$ and $\TC^0$, while highly limited in computation, also contain some problems that are not even decidable, i.e., for which there doesn't exist any exact algorithm. Thus, non-uniform classes cannot be directly compared with standard algorithmic complexity classes such as $\mathsf{P}$, $\mathsf{NP}$, etc. This motivates our second key question: \emph{Can we derive \underline{uniform} upper bounds on transformers?}}


\new{We show that one can achieve both of these goals} by making the modest assumption that all values in the transformer have $\O(\log n)$ precision (where $n$ is the number of input tokens), and, similarly, that transformer's subnetworks are computable in $\O(\log n)$ space.
Log precision is enough to represent the positional encodings at the input layer of the transformer, and to encode pointers to all other positions in the sequence at later transformer layers.
Assuming log precision across all layers captures the idea that
\new{the hidden representations contain a constant number of hidden states whose precision (16 or 32 bits) is small relative to the length of the input (2048 in GPT-3).}
On long sequences, the precision will not be enough to losslessly encode the full input sequence into a single vector. Instead, the processing of the sequence must somehow be distributed in each layer and performed in parallel.

\paragraph{Upper Bound on Transformers.} Our main contribution is proving that log-precision transformers can be simulated by \emph{uniform} constant-depth threshold circuits. Thus, \textbf{such transformers can only solve problems in uniform $\TC^0$}. \new{This characterization is strikingly weak} compared to the Turing-completeness of infinite-precision transformers. Since we believe log precision is more realistic for practical transformers than infinite precision, these results point to the conclusion that transformers are not Turing-complete in practice.

In contrast to past results, our upper bound on transformers is a \emph{uniform} circuit class, enabling direct comparison of log-precision transformers to many natural complexity classes.
These connections reveal specific problems that define the upper limits of log-precision transformers' capabilities, as discussed further in \autoref{sec:implications}.

\new{Intuitively, our upper bound says that log-precision transformers are computationally shallow, and that this shallowness can be understood to emerge from their parallelizability.
Transformers' inherent parallelism is useful for training them efficiently at massive scale, but may limit the complexity of the computations they can express.
We introduce the term \textbf{parallelism tradeoff} to capture this idea, which represents a potential fundamental weakness of the current paradigm of scaling language models.
Formally characterizing reasoning capabilities relevant to language models and understanding whether they likely fall outside upper bounds implied by the tradeoff would clarify the practical implications of this limitation of scaling.

It could also be that the limitations of parallelism are not a curse but a blessing, if they constrain the hypothesis space in a way useful for learning. We have no evidence that this is true, but mention it as an alternate interpretation of the results that could be clarified in future work.
}



\paragraph{Instruction Following and Advice Transformers.}
We also consider an instruction following setting~\citep{gpt3} where the transformer is provided the description of a task along with an input on which to execute the instruction. We construct a practically parameterizable transformer that can  execute instructions perfectly if they are provided in the form of $\TC^0$ circuits. This complements recent work that studies transformers' ability to follow other forms of instructions such as regular expressions \citep{finlayson2022instruction}.

Based on the fundamental property that transformers can correctly \emph{evaluate} any given $\TC^0$ circuit on a given input, we introduce the notion of \emph{advice transformers} akin to advice taking Turing machines. We show that transformers can recognize any (non-uniform) $\TC^0$ language if provided appropriate \new{poly-size} advice.

\paragraph{}
In summary, our findings provide new insights on both the abilities and the limitations of transformers, and bring out bounded precision, threshold computations, and parallelism as key notions for understanding the implicit computational model of transformers in practice.

\new{\paragraph{Roadmap.} Before diving into technical details, we discuss in \autoref{sec:implications} the implications of our results on both fundamental as well as practical abilities of transformers. \autoref{sec:circuits} provides a brief primer on circuits as a model of computation. It then discusses a way of serializing a circuit into a string; we later show how to generate such serializations using a resource-bounded algorithm, which is the key to proving containment of transformers in \emph{uniform} circuit classes. \autoref{sec:bounded-precision-transformers} defines our formal model of bounded-precision transformers. \autoref{sec:nonuniform-bounds} derives our first formal bound on log-precision transformers. This bound involves \emph{non-uniform} circuit families, similar in spirit to prior results in this area. \autoref{sec:uniform-bounds} proves our more technical main result: the first \emph{uniform} circuit complexity upper bound for transformers (specifically, uniform $\TC^0$). Finally, \autoref{sec:lower} provides a \emph{lower bound} on transformers, introduces the notion of an Advice Transformer, and connects these to the machine learning problems of Instruction Learning and Following.
}

\section{Implications of Our Findings}
\label{sec:implications}

Before diving into technical details, we discuss the general implications of our findings on the abilities and limitations of transformers. \new{We will focus here on our main result (\autoref{thm:uniform}), which shows that log-precision transformers are in the complexity class logspace-uniform $\TC^0$.}

\paragraph{The Parallelism~Tradeoff.} One interpretation of complexity classes such as \new{$\NC^0$, $\AC^0$, and} $\TC^0$ is sets of poly-time solvable problems that are parallelizable to a very high degree---they can be solved in parallel in \emph{constant} time with enough parallel processors. This gives some intuitive explanation of our result: log-precision transformers end up in $\TC^0$ because they were designed to be highly parallelizable. Since parallelism is an important property \new{of today's dominant paradigm of training models at} massive scale, this points to the conclusion that any massively scaled up model---transformer or otherwise---\new{will likely} obey restrictions similar to the ones derived here for log-precision transformers. \new{There is thus an important tradeoff between the massive parallelizability of today's networks and their representation power.}

\paragraph{What Transformers Can/Cannot Compute.} \new{Our result places log-precision transformers in the complexity class logspace-uniform $\TC^0$. This has immediate implications on the kinds of problems such transformers can and cannot accurately solve.

Consider any problem $X$ that is complete for a complexity class $\C$ that contains logspace-uniform $\TC^0$. By definition of completeness, every problem log-precision transformers can solve perfectly is efficiently reducible to $X$ and is thus no harder than $X$. This implies that---despite their massive size---the computation performed by such transformers is, for instance, no harder than solving basic $\mathsf{L}$-complete problems like \textbf{graph connectivity}: the problem of checking whether there is a path between two nodes in an undirected graph \citep{Lewis1982SymmetricSC,Reingold2008UndirectedCI}.

By the same token, if $\C$ is strictly larger than logspace-uniform $\TC^0$, then such transformers \emph{cannot} perfectly solve $X$.
Thus, log-precision transformers cannot perfectly solve the following reasoning problems:
\begin{compactitem}
    \item \textbf{Linear equalities}: find $\mathbf x$ s.t. $\mathbf A \mathbf x = \mathbf b$\footnote{\label{footnote:neqP}\new{Assuming logspace-uniform $\TC^0 \neq$ $\mathsf{P}$. Follows because these problems are $\mathsf{P}$-complete~\cite{p-complete-problems}.}}
    \item \textbf{Universal context-free recognition}$^{\ref{footnote:neqP},}$\footnote{\new{Takes both a \new{grammar and a string} as input and return whether the grammar generates the string.} \citet{JONES1976105} demonstrate $\mathsf{P}$-completeness.}
    \item \textbf{Propositional satisfiability} (SAT)\footnote{\new{Assuming logspace-uniform $\TC^0 \neq$ $\mathsf{NP}$. Follows because SAT is $\mathsf{NP}$-complete \citep[cf.][]{handbook-of-sat}.}} 
    \item \textbf{Horn-clause satisfiability} (HORN-SAT)$^{\ref{footnote:neqP}}$
    \item \textbf{AI planning}~\cite{Bylander1991ComplexityRF}
    \item \textbf{Permanent computation}\footnote{Assuming logspace-uniform $\TC^0 \neq$ $\mathsf{\#P}$. Follows because permanent is $\mathsf{\#P}$-complete~\citep{Valiant1979TheCO}. \citet{Allender1999ThePR} shows permanent is not in \emph{logtime}-uniform $\TC^0$.}
\end{compactitem}
This highlights the limits of practical transformers with limited-precision arithmetic, indicating that they are far from being universal or all-powerful as suggested by some prior studies.}

\new{One important caveat about these negative results is that they are asymptotic in nature---they apply for ``large enough'' input size $n$. It's possible for log-precision transformers to solve such problems easily when $n$ is small. Further, these negative results are about exact solutions, but they often also extend beyond this when formal hardness-of-approximation results are known.}

\paragraph{Limitations of Our Formal Model.} Prior formal characterizations of transformers either make unrealistically strong assumptions~\cite{perez2019on,Dehghani2019UniversalT} or place unrealistic restrictions~\cite{hahn-2020-theoretical,hao2022,merrill2022SatAttnTC0}. In contrast, we make only one assumption---namely, all intermediate values in the transformer are limited to $\O(\log n)$ bits, where $n$ is the number of input tokens. \new{We next discuss some implications of this assumption and what our findings mean for practical transformers.}

\new{As mentioned above, our bounds are asymptotic in nature and thus apply when $n$ is sufficiently large.} In practice, transformers use fixed precision at each computation node, which is more restrictive than \new{precision growing with the input sequence length $n$, as} $\O(\log n)$ bits. However, this constant \new{could be} large and thus, for relatively small $n$, our results do not rule out practical transformers solving difficult problems. Our results, however, do show that as $n$ grows sufficiently large, log-precision transformers are fundamentally limited to problems within $\TC^0$ \new{and cannot accurately solve various commonly studied problems mentioned earlier under ``What Transformers Cannot Compute''}. Extending our analysis to small $n$ will help close the gap to practice.

Our formal model is based on a binary classification view of transformers. However, our results apply directly to multi-class classification as well and can be extended to generation problems by viewing, for instance, next word prediction in NLP as a multi-class classification problem. \new{However, if the transformer decoder is allowed to condition on its previous output in a generation problem, then this would violate our formal setup.}

\subsection{\new{Potential Applications}}

\new{\paragraph{Extracting Circuits from Transformers.} \citet{elhage2021mathematical} propose extracting circuits\footnote{\new{Their sense of ``circuit'' is not exactly the formal sense we use in this paper, though the goal of capturing transformers' implicit computational mechanism is the same.}} that capture the computational structure of transformers. Our results suggest threshold circuit families are a good formalism for expressing mechanisms extracted from transformers. Constructively converting transformers to threshold circuits is beyond the scope of the current paper, although we hope to explore this in more detail in future work.}

\paragraph{Testing Separation Candidates in Complexity Theory.} \autoref{thm:uniform} also motivates a paradigm for quickly testing complexity theory conjectures. If a problem is believed to separate $\TC^0$ and $\NC^1$,
a transformer can be trained on problem instances. \new{If the transformer generalizes perfectly to harder instances than it was trained on}, this \new{gives an empirical hint} that the problem is in $\TC^0$, \new{providing evidence against the conjecture.}

\section{Circuit Computation}
\label{sec:circuits}

Let $\{0, 1\}^*$ be the set of finite binary strings.
For $x \in \{0, 1\}^*$, let $\abs{x}$ be its length.
We refer to a function from $\{0, 1\}^*$ to $\{0, 1\}^*$ as a boolean function.
Boolean functions can implement arithmetic operations if we define a semantics for binary strings as numbers.
We will treat the intermediate values in a transformer as binary strings, and the internal operations as boolean functions.

\emph{Circuits} are a model of computation for computing boolean functions of fixed-length binary strings.\footnote{For a mini-tutorial on circuit complexity theory and its relevance to transformers, see \citet{merrill2022SatAttnTC0}.} Formally, a circuit is a directed acyclic computation graph. The leaf nodes represent binary variables and their negations. The internal nodes represent functions in some set $\mathcal G$, and the directed edges represent the flow of function outputs into inputs of other functions. One or more nodes in the circuit are marked such that their value is the output of the circuit.

\begin{definition}
For a set of functions $\mathcal G$, a $\mathcal G$-circuit is a directed acyclic computation graph where the internal nodes have labels from $\mathcal G$.
\end{definition}

\paragraph{Complexity Measures.} The \emph{size} of a circuit is the total number of gates in it, including negation. The \emph{depth} of a circuit is the length of the longest path from any input node to any output node.

\paragraph{Circuit Families.}
A \emph{circuit family} generalizes a circuit to take variable-length binary strings as input. Formally, a circuit family is a sequence of circuits $C_n : \{0, 1\}^n \to \{0 , 1\}$ for $n \in \mathbb N$. A circuit family implicitly recognizes a formal language defined as follows:
\begin{definition}
A circuit family $C_n$ recognizes $L \subseteq \{0, 1\}^*$ if, for all $x \in \{0, 1\}^*$,
$C_{\abs{x}}(x) = 1$ if and only if $x \in L$.
\end{definition}

We now define classes of languages by constraining the complexity of the circuit families needed to recognize them:
\begin{definition}
Let non-uniform $\AC^0$ be the set of $L \subseteq \{0, 1\}^*$ such that $L$ is recognizable by a \new{poly-size}, constant-depth $\{\neg, \wedge, \vee\}$-circuit family.
\end{definition}

For $k \in \mathbb{N}$, a \emph{threshold gate} $\theta_{\leq k}$ takes $m$ input bits and returns whether $\sum_{i=1}^m x_i \leq k$. We define $\theta_{\geq k}$ analogously. For example, $\theta_{{\leq}3}(110011) = 0$.
\begin{definition} \label{def:tc0}
Let $\TC^0$ be the set of $L \subseteq \{0, 1\}^*$ such that $L$ is recognizable by a \new{poly-size}, constant-depth $\{\theta_{\leq k}, \theta_{\geq k}\}_{k \in \mathbb N}$-circuit.
\end{definition}
The gates $\neg$, $\wedge$, and $\vee$ are all just special cases of thresholds, so we can imagine $\TC^0$ circuits to have access to these as well. Thus, $\TC^0$ circuits can implement $\AC^0$ circuits.

\paragraph{Circuit Serialization.} We identify a circuit with its serialization in a formal language that identifies each node's label and adjacency list. We will adopt a specific grammar for concreteness, but our construction can be adapted to other string representations of circuits.

We define a circuit serialization as a \new{traversal of a circuit} ordered by some topological sort. \new{In this serialization, leaf nodes (variables) are represented by the string \texttt{X}. An internal node (gate) is represented in Polish notation by the function it computes (\texttt{AND}, \texttt{OR}, or \texttt{NOT}) followed by a list of pointers to its arguments.
Each argument $\texttt{\&1}^j$ of gate $i$ encodes (in a unary) a zero-indexed pointer to the $j$-th gate in the circuit, where $j < i$. The final node is interpreted as the circuit output.}

To serialize $\{\wedge, \vee\}$-circuits, we use the following grammar, where the $i$ parameter is passed through $\gate[i]$ nonterminals to track the index of the gate in left-to-right order:
\begin{small}
\begin{align*}
    \circuit \ & \to \ \gate[1] \; \gate[2] \; \cdots \; \gate[g] \\
    \gate[i] \ & \to \ \xx \; \mid \; \texttt{NOT} \; \argument[i] \; \mid \; \mathrm{Op} \; \argument[i]^* \\
    \argument[i] \ & \to \ \texttt{\&}\texttt{1}^j \quad \textrm{s.t.} \; j < i \\
    \mathrm{Op} \ & \to \ \texttt{AND} \; \mid \; \texttt{OR}
\end{align*}
\end{small}
In the $\argument[i]$ rule, we enforce that $j < i$ so that arguments must be pointers to already defined gates.
\new{As an example of this serialization language, the circuit for $x_1 \vee \neg x_2 \vee x_3$ is represented as}\footnote{Spaces here (and in the grammar) are added for readability. We will ignore these spaces when passing circuit serializations as inputs to a transformer in \autoref{sec:lower}.}
\begin{verbatim}
    X X X NOT &1 OR & &111 &11
\end{verbatim}
\new{
By convention (cf.~\autoref{sec:circuits}), negations in $\AC^0$ circuits are usually taken to occur at the beginning of the circuit, rather than after $\wedge$ or $\vee$ nodes.\footnote{\new{We can apply De Morgan's laws to force any $\AC^0$ circuit to have this property.}} Our serialization grammar does not enforce this property, but of course any circuit with this property can be serialized by our grammar.}

It is a bit more complicated to serialize threshold circuits.
Formally, a threshold circuit serialization is generated by the following grammar:
%
\begin{small}
\begin{align*}
    \circuit \ & \to \ \gate[1] \; \gate[2] \; \cdots \; \gate[g] \\
    \gate[i] \ & \to \ \xx \; \mid \;
    \dir \; \texttt{1}^k\texttt{0}^{m-k} \; \argument[i]^m \\
    \argument[i] \ & \to \ \texttt{\&}\texttt{1}^j \quad \textrm{s.t.} \; j < i \\
    \dir \ & \to \ \lequal \; \mid \; \gequal
\end{align*}
\end{small}

In the rewrite rule for $\gate[i]$,
$m \in \mathbb{N}$ is the arity of the gate, and $k \leq m$ is its threshold.
The span \texttt{1}$^k$ after $\dir$ can be interpreted semantically as a unary encoding of the parameter $k$ for a threshold gate, padded by \texttt{0}'s to the number of total arguments of gate $i$.
For simplicity, we imagine $\neg$ gates are represented as unary $\theta_{\leq 0}$ gates. Thus, the circuit $\theta_{\geq 1}(x_1, \neg x_2)$ would be represented as
\begin{verbatim}
    X X <= 00 &1 >= 10 & &11
\end{verbatim}
We say a threshold circuit serialization is in \emph{prefix form} if all inputs (\texttt{X}) come before all threshold gates (\texttt{<=} or \texttt{>=}), as is the case in this example.

\paragraph{Uniformity.}

\new{The circuit families we have defined above are \emph{non-uniform}, meaning that we do not enforce that the circuits processing different input sizes must be related in any way. In degenerate cases, non-uniform circuit families can solve undecidable problems\footnote{\new{Consider the unary language $1^n$ such that Turing machine $n$ (under some arbitrary enumeration) halts. This problem is in non-uniform $\AC^0$ since we can hard-code the right answer for each $n$ in $C_n$.}} because they have infinite description length, making them a physically unrealizable model of computation. Complexity theorists have thus introduced \emph{uniform} circuit families. Uniform circuit families are a realizable model of computation with relations to classes in computational complexity and formal language theory.}

Intuitively, in a uniform circuit family, the circuits for different input sizes must be ``somewhat \new{similar}'' to each other. We formalize this \citep[cf.][]{arora2009computational} by saying that there exists a resource-constrained Turing machine that maps the input $1^n$ to a serialization of circuit $C_n$.
\begin{definition}
A language $L$ is $(S(n),I(n))$-space uniformly computable by a circuit model $M$ iff there exists a Turing machine that, for all $n \geq 0$, uses $S(n)$ space to
map $1^n$ to an $M$-circuit recognizing $L$ on inputs of size $I(n)$.
\end{definition}

This notion of uniformity is more general than the standard notion in that the input size $I(n)$ is a function of the problem complexity $n$. The reason for this is that we will apply uniformity to subcomputations with different input sizes $I(n)$ within a larger computation of input size $n$. The standard notion of uniformity corresponds to $I(n) = n$.

Furthermore, we will refer to a circuit family as \emph{uniform} if it \new{is} uniformly computable with $S(n) = \O(\log n)$ \citep[cf.][]{arora2009computational}. We can define uniform versions of $\AC^0$ and $\TC^0$ by adopting the previous definitions exactly, but also enforcing uniformity.
For the rest of the paper we will clarify whether we mean the uniform or non-uniform variant of $\TC^0$ \new{when unclear from context}, since both classes \new{will come up.}


\section{Bounded-Precision Transformers}
\label{sec:bounded-precision-transformers}

A \emph{transformer} \citep{vaswani2017attention} is a neural network architecture made up of a constant number of \emph{transformer layers}. A transformer layer is a module that computes self-attention over a sequence followed by an elementwise transformation of the output vectors.

\subsection{Precision and Space}

We will assume that each transformer is resource bounded in terms of the \emph{precision} of each value it computes and, for some of our results, the \emph{space} it uses for the computation of key operations such as embedding, attention, and activation. Specifically, we will assume precision $p$, i.e., the values at all layers, as well as the outputs of all key intermediate operations in it (attention, activation, arithmetic  operators, etc.), are represented using $p$ bits. This is a realistic assumption as, in practice, today's transformers are typically limited to the 64-bit precision of the underlying hardware. Formally, we define $p$-precision as follows:
\begin{definition} \label{def:p-precision}
A $k$-ary function $f: x_1, \ldots, x_k \mapsto y$ is \emph{$p$-precision} if \new{$x_1, \ldots, x_k, y \in \{0,1\}^*$ have size at most $p$ bits}, and $f$ can be computed \new{by a $p$-space-bounded} Turing machine.
\end{definition}

This says the size of the function input and output are bounded below $p$. Similarly, the intermediate space used by the computation must also be bounded below $p$. Thus, higher precision computations cannot somehow be hidden inside $f$.

\autoref{def:p-precision} naturally applies to functions with bounded arity $k$. We will also need to define $p$ precision for the summation operator in the transformer, which adds $n$ different \new{floats of size $p$}.\footnote{\new{Our proof also goes through if the transformer weights are integers, as is sometimes done \citep{dettmers2022gptint}.}}
Adding $n$ \new{floats} can blow up the precision needed to represent their sum. For example, imagine adding the floating points $1 \cdot 2^0 + 1 \cdot 2^c$. We obtain $(2^c + 1) \cdot 2^0$, whose mantissa takes $c + 1$ bits to represent. In practice, computers do not preserve full precision in such situations: instead, small terms like $1 \cdot 2^0$ are discarded. \new{Thus, we define the transformer's addition operation $\oplus$ to be similarly approximate (and thus preserve precision); see \autoref{sec:addition}.}

\subsection{Transformer Definition}

\new{
\subsection{Attention Heads} \label{sec:head}

The core building block of a transformer is an attention head. We define this at a high level of abstraction as follows:

\begin{definition} \label{def:head}
A $p$-precision attention head is specified by a binary $p$-precision \emph{similarity} function $s : \{0, 1\}^p \times \{0, 1\}^p \to \{0, 1\}^p$.
\end{definition}

Let $\mathbf h_1, \ldots, \mathbf h_n \in \{0, 1\}^p$ be the input sequence to a $p$-precision attention head\new{, and let $\oplus$ be approximate floating-point addition (\autoref{sec:addition}).}

\begin{definition} \label{def:head-computation}
For all $\ell \geq 0$,
a $p$-precision attention head $H^{\ell+1}_h$ computes a vector $\mathbf a^{\ell+1}_{ih} \in \{0, 1\}^p$ via
\begin{equation*}
    \mathbf a^{\ell+1}_{ih} = \bigoplus_{j=1}^n \frac{s(\mathbf h^\ell_i, \mathbf h^\ell_j)}{Z_i} \cdot \mathbf h^\ell_j ,
\end{equation*}
where $Z_i = \bigoplus_{j=1}^n s(\mathbf h^\ell_i, \mathbf h^\ell_j)$.
\end{definition}

Standard transformer attention heads \citep{vaswani2017attention} are a special case of this definition where $s$ is scaled dot-product similarity between keys and queries. Standard transformers also have a linear or affine \new{value} function applied to each $\mathbf h^\ell_j$ in the sum over $j$. By its affineness, the value function can, without loss of generality, be removed from the attention head and considered to be part of the transformer layer (i.e., applied to the output of the attention head).

\subsection{Transformer Layers}

A $p$-precision transformer layer is then a tuple of heads and a function $f$ used to combine them.

\begin{definition}[$p$-precision transformer layer] \label{def:layer}
A $p$-precision transformer layer is a tuple $L^{\ell+1} = \langle H_1, \cdots, H_k, f \rangle$, where each $H_h$ is an attention head and $f : \left( \{0, 1\}^p \right)^k \times \{0, 1\}^p \to \{0, 1\}^p$ is a $p$-precision \emph{activation} function.
\end{definition}

A $p$-precision transformer layer can be understood to define a sequence of vectors $\mathbf h^{\ell+1}_1, \ldots, \mathbf h^{\ell + 1}_n$ in terms of an input sequence of vectors $\mathbf h^{\ell}_1, \ldots, \mathbf h^{\ell}_n$ (coming from the previous layer in the transformer) by first computing $k$ attention heads in parallel and then combining their output using $f$.
The first $k$ inputs to $f$ will correspond to the attention head outputs, and the additional input is the original input from the previous layer.
Recall that $\mathbf a^{\ell+1}_{ih}$ is the output of head $H^{\ell+1}_{ih}$ on input $\mathbf h^\ell$ at position $i$.
The function computed by a transformer layer can be described formally as follows.

\begin{definition}[Transformer layer computation]
\label{def:layer-computation}
For $\ell \geq 0$, a $p$-precision transformer layer $L^{\ell+1}$ recurrently computes the output sequence $\mathbf h_1^{\ell+1}, \ldots, \mathbf h_n^{\ell+1}$ as a function of the inputs $\mathbf h_1^{\ell}, \ldots, \mathbf h_n^\ell$, where, for $1 \leq i \leq n$, the $i$-th component is computed according to
\begin{equation*}
    \mathbf h^{\ell+1}_i = f(\mathbf a^{\ell+1}_{i1}, \ldots, \mathbf a^{\ell+1}_{ik}, \mathbf h_i^{\ell}) .
\end{equation*}
\end{definition}

$f$ can be understood to encapsulate layernorm, residual connections, and the feedforward sublayer of a standard transformer \citep{vaswani2017attention}.
$\mathbf h_i^{\ell}$ is given to $f$ to allow residual connections.
As mentioned in \autoref{sec:head}, $f$ can also encapsulate the value function for each head.
}


\new{
\subsection{Transformer Encoder}}

Finally, we define a transformer of depth $d$ as a cascade of $d$ transformer layers:

\begin{definition}[$p$-precision transformer]
A $p$-precision transformer over alphabet $\Sigma$ is a pair consisting of a $p$-precision position embedding function\footnote{To apply the normal notion of $p$-precision to inputs outside $\{0, 1\}^*$, we imagine elements of $\Sigma$ are encoded as integers $\leq \abs{\Sigma}$ in binary, and natural numbers are represented as integers $\leq n$. Thus, we assume $\log \abs{\Sigma} + \log n \leq p$.} $\phi : \Sigma \times \mathbb N \to \{0, 1\}^p$ and a $d$-tuple of $p$-precision transformer layers \new{$\langle L^1, \ldots, L^d \rangle$}.
\end{definition}

For a position embedding function $\phi$ and $w \in \Sigma^n$, let $\phi(w)$ be the position-wise broadcasted embedding of $w$: for $1 \leq i \leq n$, $\phi_i(w) \triangleq \phi(w_i, i) $.

\begin{definition}[Transformer computation]
A transformer $\left( \phi, \langle L^1, \cdots L^d \rangle \right)$ computes the following function of a string $w \in \Sigma^*$:
\begin{equation*}
    T(w) = (L^d \circ L^{d-1} \circ \dots \circ L^1)(\phi(w)) .
\end{equation*}
\end{definition}

We will use $n$ to denote the length of $w$,
and take the transformer's depth $d$ to be fixed w.r.t.\ $n$.

The \textbf{input} to the transformer can thus be represented with $N = n \log \abs{\Sigma}$ bits using a binary encoding for the vocabulary. The circuits we construct in subsequent sections to simulate transformers will also have input size $N$. We will assume transformers have \textbf{log-precision} relative to the size of the input, specifically $\O(\log N)$-precision. Since $\abs{\Sigma}$ is fixed (typically 30000 in practice), we will think in terms of $\O(\log n)$-precision. Thus, by \autoref{def:p-precision}, all of the intermediate functions of such transformers are computable in $\O(\log n)$ space and output (at most) these many bits. Note that this is enough precision to represent positional encodings and for each position to point to a constant number of other values, but not enough precision for non-lossy pooling of the entire input into a single value.


\paragraph{Relationship to Practical Transformers.} Our \new{log-precision} transformers do not enforce that \new{$s$  (\autoref{def:head}) and $f$ (\autoref{def:layer}) follow the transformer structure}. However, a feedforward net whose primitive operations (e.g., scalar multiplication) are defined over $\O(\log n)$-size numbers can be computed in $\O(\log n)$ space. Thus, bounded-precision practical transformers are a special case of our \new{log-precision transformers}. This makes our setup appropriate for proving upper bounds on transformers, which is our main contribution.


\section{Log-Precision Transformers as Non-Uniform Threshold Circuits}
\label{sec:nonuniform-bounds}

We first show that log-precision transformers can be simulated by \emph{non-uniform} threshold circuits, before presenting the more technical \emph{uniform} version of the results in \S\ref{sec:uniform-bounds}.
The initial non-uniform result extends the findings of \citet{merrill2022SatAttnTC0}, who showed that \emph{saturated} attention transformers\footnote{Saturated attention is uniform attention over a subset of the prior layer nodes.} can be simulated in $\TC^0$. Here, we remove the simplifying saturated attention assumption and other restrictions on the underlying datatype. Instead, we show that our log-precision assumption is enough to prove that a transformer can be simulated in $\TC^0$ with any attention function.

\citeauthor{hao2022} observed that any boolean function of $\O(\log n)$ bits can be computed by a poly$(n)$ size circuit. We extend this to $m$-bit outputs, which is both more convenient and more efficient than constructing $m$ separate boolean circuits:

\begin{lemma}[Extended from \citealp{hao2022}] 
\label{lem:hao-extended}
Let $f : \{0, 1\}^* \to \{0, 1\}^m$ be a function. For all $c \in \mathbb{R}^+$ and $n \in \mathbb{N}$, there exists an AND/OR circuit of size at most $n^c + c \log n + m$ and depth $3$ that computes $f$ on inputs of size $c \log n$.
\end{lemma}

\begin{proof}
Like \citet{hao2022}, we construct a circuit using a DNF representation of $f$ on inputs of size $c \log n$, except we use a combined DNF representation for all output bits of $f$. The DNF formula has at most $2^{c \log n} = n^c$ terms. The circuit has a NOT gate for each input bit, an AND gate for each DNF term, and, for each of the $m$ output bits, an OR gate combining the outputs of those AND gates (i.e., DNF terms) for which that bit is $1$.
\end{proof}


We now use \autoref{lem:hao-extended} to prove the following non-uniform result. We note that the proof goes through even if the notion of $p$-precision (\autoref{def:p-precision}) is relaxed to not require computability in space $p$. This requirement will, however, become important for our subsequent result in \autoref{sec:uniform-bounds}.

\begin{theorem}[Non-uniform]
\label{thm:nonuniform}
Any $c \log n$-precision depth-$d$ transformer operating on inputs in $\Sigma^n$ can be simulated by a threshold circuit family of depth $3 + (9 + 2\dplus) d$.
\end{theorem}

\begin{proof}
Let $w \in \Sigma^n$ be the input of a $c \log n$-precision transformer.
We show by induction that we can construct a composition of \new{constant-depth, poly-size} threshold circuits to compute each layer of this transformer. Thus, any constant-depth transformer will be computable by a constant-depth threshold circuit.

In the base case of layer $0$ and token $i$, we construct gates representing the constant $i$ encoded in binary. We can then compute $\mathbf h_i^0 = \phi(w_i, i)$ using \autoref{lem:hao-extended}, yielding a \new{poly-size} depth-3 circuit. 

In the inductive case of computing layer $\mathbf h_i^{\ell+1}$ for $1 \leq \ell+1 \leq d$, we note that each vector output of layer $\mathbf h_i^\ell$ has size (at most) $c \log n$ bits because of the log-precision assumption.

\new{We first fix a head $\mathbf a^{\ell+1}_{ik}$ (\autoref{def:head-computation}) to simulate.}
Applying \autoref{lem:hao-extended}, \new{we can compute $s(\mathbf h_i^\ell, \mathbf h_j^\ell)$} with a \new{poly-size} depth-3 circuit, in parallel for all $j$.
\new{Since $n$ floats with $c \log n$ precision can be approximately added in $\TC^0$ (\autoref{sec:addition}),}
we can construct a $\TC^0$ circuit of depth $\dplus$ to compute $Z_j$. Since $s(\mathbf h_i^\ell, \mathbf h_j^\ell), Z_i$, and $\mathbf h_i^\ell$ all have $c \log n$ bits, we can compute $\frac{s(\mathbf h_i^\ell, \mathbf h_j^\ell)}{Z_i} \mathbf h_j^\ell$ with a \new{poly-size} depth-3 circuit;\footnote{This may seem counterintuitive since multiplication of two $n$-precision numbers is outside $\AC^0$. However, here we leverage the fact that the precision is $c \log n$.} we do this in parallel for all $j$. Next, we \new{again use the fact that approximate addition of $n$ floats is in $\TC^0$ to} compute \new{$\mathbf a_{ih}^{\ell+1}$ as the approximate sum} over $j$ \new{with a} depth-$\dplus$ circuit.

\new{We now simulate a layer $\mathbf h^{\ell+1}_i$ (\autoref{def:layer-computation}) in terms of its constituent heads.}
\new{Since all} arguments of $g$ have size $c \log n$, we apply \autoref{lem:hao-extended} to compute $g$ with a \new{poly-size} depth-3 circuit, yielding $\mathbf h_i^{\ell+1}$. We repeat this in parallel for all $i$. This completes the inductive step new \new{to compute} all values in the $\ell+1$-st layer with a circuit depth of $9 + 2\dplus$.

Aggregating the circuit over all $d$ layers, the overall circuit depth is $3 + (9 + 2\dplus) d$.
\end{proof}

\begin{corollary}[Non-uniform]
\label{cor:nonuniform}
Any log-precision transformer can be simulated by \new{a non-uniform $\TC^0$ circuit family}.\footnote{\new{Here, a $\TC^0$ circuit family is a constant-depth, poly-size circuit family computing some function $\{0, 1\}^* \to \{0, 1\}^*$. While we define $\TC^0$ for decision problems in \autoref{def:tc0}, it is standard and well-defined to extend the same term to refer to circuit families computing functions as well \citep{hesse2001division}.}}
\end{corollary}

\section{Log-Precision Transformers as \emph{Uniform} Threshold Circuits}
\label{sec:uniform-bounds}

We will now extend the argument from the last section to show that $\O(\log n)$-precision transformers can be simulated by uniform constant-depth threshold circuits by capitalizing on the assumption that $\phi, s,$ and $f$ are log-precision, and thus can be computed in $\O(\log n)$ space.
The overall proof idea is similar, but due to the uniformity condition, the proof becomes substantially more technical. We must not just show the existence of a threshold circuit family computing a transformer, but also show that this circuit family can be generated by a log-space Turing machine.



We first extend \autoref{lem:hao-extended} to respect uniformity:

\begin{lemma} \label{lem:ff} \label{lem:eval}
Let $f : \{0, 1\}^* \to \{0, 1\}^m$ be a linear-space computable function. There exists a Turing machine that, for all $n \in \mathbb{N}$ and \new{$c \in \mathbb{R}^+$,} uses at most \new{$c \log n + \log m$} space to map input $1^n$ to a circuit of size at most $n^c + c \log n + m$ and depth $3$ that computes $f$ on inputs of size \new{at most} $c \log n$.
\end{lemma}

\begin{proof}
We give the proof in the form of an algorithm to construct a circuit as a function of $n$ and then justify its correctness and space complexity.

\paragraph{} \underline{Algorithm.}
We first print $2c \log n$ nodes representing unnegated and negated input nodes.\footnote{We ignore the initial unnegated input nodes when considering the size of the circuit.}

Now, we need to show how to construct nodes corresponding to $n^c$ DNF terms.
To this end, we loop over all possible inputs $x \in \{0, 1\}^{c \log n}$ by maintaining the $c \log n$ bit binary representation of $x$ (initialized with $0^{c \log n}$) and incrementing it by $1$ at each step of the loop. We create a new $\wedge$ node $i$ with $c \log n$ arguments, defined as follows. For $j \in [c \log n]$, we create an argument pointer to (unnegated) node $j$ if $x_j = 1$ and to (negated) node $c \log n + j$ otherwise.

Now, we construct nodes computing each of the $m$ output nodes. We loop over $k \in [m]$, constructing a single node for each $k$. We loop over all $x \in \{0, 1\}^{c \log n}$ analogously above to construct a list of arguments. By our linear-space computability assumption and because $x$ has $c \log n$ bits, we can compute $f(x)$ as a subroutine in $\O(\log n)$-space to obtain $f_k(x)$.
If $f_k(x) = 1$, we print node $2c \log n + j$ as an argument of node $k$.

\paragraph{}
\underline{Correctness.}
We show that this Turing machine maps input $n$ to a serialized circuit computing $f$ on inputs of size $n$.
\new{The first layer simply produces unnegated and negated input values.} \new{The second layer then} produce all possible DNF terms. \new{Finally,} node $k$ of the \new{third layer} computes the disjunction over all terms $x$ such that $f_k(x) = 1$. Thus, node $k$ of the \new{third layer} computes $f_k$.

\paragraph{}
\underline{Log Space.}
To complete the proof, we justify that $M$ uses $\O(\log n \new{+ \log m})$ space. Looping over $x \in \{0, 1\}^{c \log n}$ is accomplished by treating $x$ as a binary number initialized to $0$ and incrementing it \new{at} each step.
Thus, the loop pointer for building the DNF terms takes $c \log n$ space to store.
For building the \new{$m$} output nodes, we maintain a similar loop pointer as well as an index $k \leq m$, taking \new{$c \log n + \log m$} space. 
Thus, the \new{overall} algorithm uses \new{$c\log n + \log m$} space.

\paragraph{}
Thus, $M$ uses \new{$c \log n + \log m$} space to map $1^n$ to a circuit of size at most $n^c + c \log n + m$ and depth $3$ that computes $f$ on size $c \log n$ inputs.
\end{proof}


We can leverage this lemma to derive the \emph{uniform} analog of \autoref{thm:nonuniform}, as follows.

\begin{theorem}[Uniform, main result]
\label{thm:uniform}
Any $c \log n$-precision depth-$d$ transformer operating on inputs in $\Sigma^n$ can be simulated by a logspace-uniform threshold circuit family of depth $3 + (9 + 2\dplus) d$.
\end{theorem}

\begin{proof}
We will provide a proof by induction over transformer layers $\ell$ that there is a Turing machine $M$ operating in $\O(\log n)$ space that, on input $1^n$, outputs a circuit that simulates the transformer's computation on inputs of size $n$. This circuit is identical to the one in the proof of \autoref{thm:nonuniform}, and thus has the same circuit depth.

In the base case, we use log space to track a counter maintaining the current token $i$ (between $1$ and $n$) throughout the circuit construction. We construct gates encoding the constant $i$ in binary. We can then apply \autoref{lem:ff} to construct a Turing machine that maps $1^n$ to a constant-depth threshold circuit computing $\mathbf h^0_i = \phi(w_i, i)$.

In the inductive case, we assume we can output in $\O(\log n)$ space a circuit computing every value \new{$\mathbf h^\ell_i$} in the previous layer $\ell$. We will show that we can, in $\O(\log n)$ space, now output a circuit computing every value in layer $\ell + 1$.

\new{As in \autoref{thm:nonuniform}, we first fix a head $\mathbf a_{ih}^{\ell+1}$ to simulate.
Recall (\autoref{def:head-computation}) that
\begin{equation*}
    \mathbf a^{\ell+1}_{ih} = \bigoplus_{j=1}^n \frac{s(\mathbf h^\ell_i, \mathbf h^\ell_j)}{Z_i} \cdot \mathbf h^\ell_j .
\end{equation*}
}
By \autoref{lem:ff}, we can generate a depth-$3$ circuit of size at most $z = n^{c'} + c' \log n + 1$, where $c' = 2c$ (since the input to $f$ is of size $2c \log n$) that computes $s(\mathbf h^\ell_i, \mathbf h^\ell_j)$ for specific $i, j$. We do this sequentially for $1 \leq j \leq n$ \new{and $1 \leq h \leq k$}, padding each circuit with unused nodes so that each one has size exactly $z$, and the $z$-th node corresponds to the output.
Thus, the indices of the output nodes for each of the columns will be
\new{$w_\ell + z(jk + h)$}
for $1 \leq j \leq n$, where $w_\ell$ is the index of the last output node $\mathbf h^\ell_n$ of the previous layer.

At this point, we use the fact that for $p = c \log n$, the $p$-precision approximate sum of $n$ $p$-precision numbers can be computed by a uniform threshold circuit \new{(\autoref{sec:addition})}. We can thus use a Turing machine as a sub-routine to generate, on input $1^n$, a \new{$k$ threshold circuits, where each has} size $z'$ that computes an $\oplus$ gate over $n$ items of precision $p$ each. We set the inputs of \new{circuit $h$} to be nodes
\new{$w_\ell + z(jk + h)$} for $1 \leq j \leq n$. By construction, this yields the \new{normalizing constants} $Z_i = \bigoplus_{j=1}^n s(\mathbf h^\ell_i, \mathbf h^\ell_j)$, whose value is located at the node at index \new{$w_\ell + znk + z'$ for head $h$}.

Using $p$-precision arithmetic operator circuits, we can now also generate a circuit to compute $\frac{s(\mathbf h^\ell_i, \mathbf h^\ell_j)}{Z_i} \mathbf h^\ell_j$ for each $1 \leq j \leq n$ and \new{$1 \leq h \leq k$}, by using index
\new{$w_\ell + z(jk + h)$}
as before for the value of $s(\mathbf h^\ell_i, \mathbf h^\ell_j)$ and index \new{$w_\ell + znk + z'h$} for the \new{normalizing constant $Z_i$ of head $h$}. Here too we use circuits of identical size $z''$, making \new{$w_\ell + k(zn + z' + z''i)$} the index of the output nodes of these $n$ circuits. Next, we again employ a $\oplus$ circuit of size $z'$, similar to the computation of $Z_i$, to compute the sum of these $n$ values. Finally, we compute $h^{\ell+1}_i$ by applying $f$ via \autoref{lem:ff}.

Note that this requires keeping only $\ell, i,$ and $n$ in memory, each of which takes $\O(\log n)$ bits. 

We repeat this process for all $1 \leq i \leq n$ to compute the entire $\ell+1$ layer, which finishes the inductive step: if we can output a circuit computing layer $\ell$ in $\O(\log n)$ space, then we can do the same for layer $\ell + 1$.
\end{proof}


\new{Because the depth derived in \autoref{thm:uniform} is constant with respect to $n$, it follows that:}

\begin{corollary}[Uniform, main result]
\label{cor:uniform}
Any log-precision transformer can be simulated \new{by a uniform $\TC^0$ circuit family}.
\end{corollary}

\section{Lower Bounds for Instruction Following and Advice Transformers}
\label{sec:lower}

\new{So far, we have shown that uniform $\TC^0$ is an upper bound for log-precision transformers. Is this upper bound tight, i.e., also a lower bound? While we do not answer this question here, we address a related question as a first step: we construct a transformer that can evaluate $\TC^0$ circuits on binary inputs, showing that transformers can compute any $\TC^0$ function when their input is augmented with the right ``instructions''. }

\new{More formally, we} consider the \textbf{Circuit Value Problem (CVP)}~\citep{ladner1975circuit}, also referred to as the Circuit Evaluation Problem, where the input is a boolean circuit $C$ and a string $x \in \{0, 1\}^n$, and the task is to return the value of $C(x) \in \{0, 1\}$. This problem is known to be complete for the class $\mathsf{P}$ under $\AC^0$ reductions~\citep{ladner1975circuit}. We will assume $C$ is serialized as described in \autoref{sec:circuits} and prove that log-precision transformers can evaluate any $\TC^0$ circuit. Note that this is an extension of the typical CVP since the circuit has threshold gates, not just standard AND/OR gates.

It is known that LSTMs cannot evaluate boolean formulae \citep{merrill2021counter}, a special case of the CVP. In contrast, we show that transformers can.

To demonstrate the practicality of our lower bound construction, we will not just prove the existence of transformers that can evaluate $\TC^0$ circuits
but also specify concrete choices for the positional embedding scheme and the class of attention functions that are sufficient to do so.

\paragraph{Fractional Positional Embeddings.}
For a vector $\mathbf x$ and scalar $y$, let $\langle \mathbf x, y \rangle$ be the vector \new{appending $y$ onto $\mathbf x$}.\footnote{\new{I.e., $\langle \mathbf x, y \rangle_i = x_i$ for $1 \leq i \leq \abs{\mathbf x}$, and $y$ if $i = \abs{\mathbf x} + 1$.}}
For $\sigma \in \Sigma$, let $v(\sigma)$ be the one-hot embedding of $\sigma$ into $\mathbb{R}^{\abs{\Sigma}}$.
For $w \in \Sigma^*$ and $i \geq 1$, the fractional positional embedding at token $i$ is
\begin{equation*}
    \phi(w_i, i) = \langle v(w_i), i / n \rangle .
\end{equation*}

\paragraph{Saturated Attention.} We imagine $f(\mathbf h_i^\ell, \mathbf h_j^\ell)$ is computed via saturated attention \citep[cf.][]{merrill2022SatAttnTC0}, which provides a simple model of the types of attention we can expect to be learned in transformers \citep{merrill2020parameter}. First, queries are computed as $\mathbf q_i = \mathbf Q \mathbf h_i^\ell$, and then keys $\mathbf k_j = \mathbf K \mathbf h_j^\ell$ Define the \new{dot-product} attention score $\sigma_{ij} = \mathbf q_i^\top \mathbf k_j$. We can then define saturated attention as
\begin{equation*}
    s(\mathbf h_i^\ell, \mathbf h_j^\ell)
    =
    \begin{cases}
    1 & \textrm{if} \; \sigma_{ij} = \max_k \sigma_{ik} \\
    0 & \textrm{otherwise.}
    \end{cases}
\end{equation*}
After normalization, saturated attention creates a distribution that is uniform over a subset of positions. Thus, it is capable of parameterizing hard attention, uniform attention over the full sequence, and various attention patterns in between.

\paragraph{Simple Pooling Functions.} For simplicity, we assume pooling functions $f$ are thresholded linear functions of their inputs. Thus, they could be implemented by a feedforward neural net. Without loss of generality, we let attention heads have a value function, which can be folded into the pooling function from the last layer \new{(see \autoref{sec:bounded-precision-transformers}).}

\paragraph{Terminology.} We use \emph{input node} to mean a token of type \texttt{X} and \emph{gate node} to mean a token of type $\mathrm{Dir}$. We call a token of type \texttt{\&} an \emph{argument}.

\paragraph{} We are now ready to present the main result. Our construction below is specific to circuits serialized in prefix form (see \S\ref{sec:circuits}), but it can be extended to other serializations as well.


\begin{lemma}
\label{lemma:circuit-eval}
For all $d$,
there exists a transformer with fractional positional embeddings, saturated attention, thresholded linear pooling functions, and depth $2d$ that,
for any threshold circuit $C$ of depth $d$ serialized in prefix form,
maps input $\langle C, x \rangle$ to the value $C(x)$.
\end{lemma}

\begin{proof}
We will construct a pair of two transformer layers that evaluate all the nodes at depth $\ell$ in the threshold circuit, for any $\ell$. It follows that a transformer of depth $2d$ can compute the value $C(x)$.

\paragraph{}
\underline{Base Case: Input Nodes}. We use an attention layer to attend uniformly over all positions with value returns $1$ if $w_i = \texttt{X}$ and $0$ otherwise. This head computes $\#(\texttt{X}) / n$, where $\#(\texttt{X})$ is the number of occurrences of $\texttt{X}$ in $w$. A second layer, then, at input node $i$, computes the positional embedding of the token representing input value $x_i$:
\begin{equation*}
    \frac{1 - \#(\texttt{X}) + i}{n} .
\end{equation*}
We attend to this position to retrieve $x_i$.
After these layers, each input node $i$ stores its value $x_i$.

We also use the base-case layers to construct an attention head that, at the $i$-th node, counts the fraction of tokens (out of $n$) that are nodes to the left of the current node. Thus, the column corresponding to node $i$ stores the value $i/n$.

At each gate node $i$, we use two more attention heads to find the index of the next $\texttt{\&}$ to the right and then count the fraction of tokens before it that are \texttt{1}. This head thus computes $k_i/m_i$ where $k_i$ is the threshold value of gate $i$ and $m_i$ is its arity.


Finally, using the first attention layer, we have each \texttt{1} node attend to the first argument symbol \texttt{\&} to its left and retrieve its index $p/n$. Then, in the second attention layer, each argument attends uniformly over all nodes with values $p/n$. The net effect is for each argument to store $j/n$, i.e., the pointer it is encoding in unary as \texttt{\&1}$^j$.

\paragraph{}
\underline{Inductive Case: Gate Nodes}. By our inductive assumption over prior layers, all tokens corresponding to circuit nodes at depth $\leq \ell$ contain their appropriate value. We now construct $2$ transformer layers to evaluate gate nodes at depth $\ell + 1$.

In the first attention layer, each argument token attends to the closest gate node $i$ to its left, which is the gate it belongs to. Recall from the base case that argument token \texttt{\&} already stores $j/n$, where $j$ is the pointer value it encodes. Each argument token now attends with query $j/n$ to retrieve from node $j$ its already computed value.

The second attention layer applies at gate nodes, not arguments. At gate $i$ of arity $m_i$, we set the attention $s(i, j)$ to indicate whether argument $j$ belongs to gate node $i$, which holds for exactly $m_i$ arguments. We set the attention value at argument $j$ to be the binary value of node $j$, which was retrieved in the previous paragraph. Thus, the attention head computes $c_i/m_i$, where $c_i$ is the number of arguments of node $i$ that are $1$. We repeat this for all gate nodes.

At this point, we have both the count of true inputs to gate node $i$ ($c_i/m_i$) and, from the base case, the threshold parameter of gate $i$ ($k_i/m_i$). Thresholding $(c_i - k_i)/m_i$ at $0$ allows us to decide, based on whether $\dir$ is $\lequal$ or $\gequal$, whether the current gate node should output a $0$ or a $1$. Repeating this for all gates at layer $\ell+1$ completes the inductive step: we can evaluate all gate nodes in this layer.
\end{proof}


\begin{theorem}
\label{thm:circuit-eval}
Depth-$2d$ transformers can solve CVP for depth-$d$ $\TC^0$ circuits.
\end{theorem}

\subsection{Instruction Following}
\label{subsec:instruction-following}

CVP is closely related to \emph{instruction learning}~\citep{gpt3} and \emph{instruction following} tasks~\citep{finlayson2022instruction}. \new{The latter task setup provides} a transformer two inputs\new{:} a regular expression $r$ as an ``instruction'', and \new{$z \in \{0, 1\}^*$}. \new{The goal of the task is to} return whether $z$ belongs to the regular language represented by $r$. Viewed from this lens, the circuit evaluation setup asks: \emph{Can transformers follow instructions provided in the form of a circuit?} As discussed below, our result says the answer is \emph{yes} for all constant depth threshold circuits. This, to the best of our knowledge, provides the first non-trivial lower bound for transformers in the instruction learning setting.

Formally, an instruction $I$ is any description\footnote{\new{Formally, a function description is a fixed-size program to compute that function under some model of computation.}} of a function $f_I$ of $\{0,1\}^*$. We say a transformer correctly follows an instruction $I$ if, for all $x \in \{0,1\}^*$, it correctly computes $f_I(x)$ on input $\langle I, x \rangle$. A non-uniform instruction description is a family of length-specific descriptions $\{I_n\}_{n=1}^\infty$. We say a transformer correctly follows a non-uniform instruction family $\{I_n\}$ if, for all $n$ and all $x \in \{0,1\}^n$, it correctly computes $f_I(x)$ on input $\langle I_n, x \rangle$. The non-uniform description $\{I_n\}$ may take any form. When it forms a $\TC^0$ circuit family, we refer to it as a $\TC^0$ instruction description. \new{Since \autoref{thm:circuit-eval} constructs a transformer that can evaluate any $\TC^0$ circuit, it follows that:}

\begin{corollary}
\label{cor:instruction-following}
There exists a depth-$2d$ transformer that can correctly follow any depth-$d$ $\TC^0$ instruction description.
\end{corollary}

Thus, transformers with simple position embeddings, attention, and pooling functions can simulate any instruction provided in the form of a $\TC^0$ circuit. We note that while it is unknown whether the class of regular languages, considered by \citet{finlayson2022instruction}, is contained in $\TC^0$, the other side is known: there \emph{are} problems computable by $\TC^0$ circuits that are not computable by a regular language. These include problems involving counting and arithmetic, which are beyond regular languages. Our results thus expand the known kinds of instructions transformers are able to follow, at least with hand-constructed weights.

\subsection{Advice Transformers}
\label{subsec:advice}

We can also view circuit evaluation abilities of transformers (\autoref{lemma:circuit-eval}) from the lens of \emph{advice taking Turing machines} which, in addition to their usual input, are also provided an input length  dependent (but input independent) advice string. For instance, $\mathsf{P} / \mathsf{poly}$ is the class of problems decidable in polynomial time when the Turing machine is given an advice string of size polynomial in the input length \citep[cf.][]{arora2009computational}.

In the same vein, let $\mathsf{T} / \mathsf{poly}$ be the class of log-precision, constant-depth transformers with polynomial advice strings. In other words, on an input of size $n$, we allow the transformer to receive an additional $\mathrm{poly}(n)$ bits of input that cannot depend on the standard input. Now let $\{C_n\}_{n=1}^\infty$ be a circuit family demonstrating that a problem is in non-uniform $\TC^0$. Then, by passing the description of $C_n$ as advice for input length $n$, it immediately follows from \autoref{lemma:circuit-eval} that advice transformers can simulate non-uniform $\TC^0$:

\begin{corollary}
\label{cor:advice}
Non-uniform $\TC^0 \subseteq \mathsf{T} / \mathsf{poly}$ .
\end{corollary}

Since non-uniform $\TC^0$ even contains some undecidable languages \citep[][Claim 6.8]{arora2009computational}, $\mathsf{T} / \mathsf{poly}$ is clearly a very powerful class and a strict superset of $\mathsf{T}$, the class of decision problems recognized by transformers (which are all decidable). Thus, a problem in $\mathsf{T} / \mathsf{poly}$ cannot always be solved by a transformer on its own. However, if given a description of \emph{how} to do so (``advice'') in the form of a $\TC^0$ circuit, our result shows that a transformer \emph{could} solve that problem.

\section{Conclusion}

Answering two open questions from \citet{merrill2022SatAttnTC0}, we prove log-precision transformers with any (including soft) attention can be simulated by \emph{uniform} constant-depth threshold circuits. This establishes \emph{thresholded addition} as a fundamental operation for understanding the computational model of transformers: any log-precision transformer can be re-expressed as a polynomial number of threshold gates stacked to a constant depth. This result also establishes potential limits on the computational power of log-precision transformers; e.g., if $\mathsf L \subset \mathsf P$, transformers cannot compute all poly-time functions. They are certainly very far from being universal. The intuition at the heart of this result is that forcing a model to be highly parallelizable likely sacrifices its expressiveness. Since parallelism seems essential to pretraining any massive model at scale, any large language model---transformer or otherwise---may suffer from a similar tradeoff.


\section*{Acknowledgments}
The authors are grateful for the valuable feedback from the anonymous reviewers and the TACL action editor Dan Gildea. They also thank Paul Beame and colleagues at AI2 including Kyle Richardson, Michal Guerquin, Peter Clark, Tushar Khot,
and especially Matthew Finlayson, whose empirical findings about instruction learning inspired \autoref{sec:lower}.
Feedback from Sam Bowman, Arya McCarthy, Roma Patel, and Lena Strobl, and discussions with the FLaNN, ML for Code (MILA), and Foundations of Language Processing (Umeå) research groups helped improve earlier drafts.
The authors also appreciate Rahul Santhanam's feedback.
This work was funded in part by NSF award 1922658. William Merrill was supported by an NSF graduate research fellowship
and by AI2.

\bibliography{references}

\begin{thebibliography}{27}
\expandafter\ifx\csname natexlab\endcsname\relax\def\natexlab#1{#1}\fi

\bibitem[{Allender(1999)}]{Allender1999ThePR}
Eric Allender. 1999.
\newblock The permanent requires large uniform threshold circuits.
\newblock \emph{Chicago Journal of Theoretical Computer Science}.

\bibitem[{Arora and Barak(2009)}]{arora2009computational}
Sanjeev Arora and Boaz Barak. 2009.
\newblock \href
  {https://books.google.com/books/about/Computational_Complexity.html?id=8Wjqvsoo48MC}
  {\emph{Computational Complexity: A Modern Approach}}.
\newblock Cambridge University Press.

\bibitem[{Biere et~al.(2009)Biere, Heule, van Maaren, and
  Walsh}]{handbook-of-sat}
Arin Biere, Marijn Heule, Hans van Maaren, and Toby Walsh. 2009.
\newblock \emph{Handbook of Satisfiability: Volume 185 Frontiers in Artificial
  Intelligence and Applications}.
\newblock IOS Press.

\bibitem[{Brown et~al.(2020)Brown, Mann, Ryder, Subbiah, Kaplan, Dhariwal,
  Neelakantan, Shyam, Sastry, Askell, Agarwal, Herbert-Voss, Krueger, Henighan,
  Child, Ramesh, Ziegler, Wu, Winter, Hesse, Chen, Sigler, Litwin, Gray, Chess,
  Clark, Berner, McCandlish, Radford, Sutskever, and Amodei}]{gpt3}
Tom Brown, Benjamin Mann, Nick Ryder, Melanie Subbiah, Jared~D Kaplan, Prafulla
  Dhariwal, Arvind Neelakantan, Pranav Shyam, Girish Sastry, Amanda Askell,
  Sandhini Agarwal, Ariel Herbert-Voss, Gretchen Krueger, Tom Henighan, Rewon
  Child, Aditya Ramesh, Daniel Ziegler, Jeffrey Wu, Clemens Winter, Chris
  Hesse, Mark Chen, Eric Sigler, Mateusz Litwin, Scott Gray, Benjamin Chess,
  Jack Clark, Christopher Berner, Sam McCandlish, Alec Radford, Ilya Sutskever,
  and Dario Amodei. 2020.
\newblock \href
  {https://proceedings.neurips.cc/paper/2020/file/1457c0d6bfcb4967418bfb8ac142f64a-Paper.pdf}
  {Language models are few-shot learners}.
\newblock In \emph{Advances in Neural Information Processing Systems},
  volume~33, pages 1877--1901. Curran Associates, Inc.

\bibitem[{Bylander(1991)}]{Bylander1991ComplexityRF}
Tom Bylander. 1991.
\newblock Complexity results for planning.
\newblock In \emph{Proceedings of the International Joint Conference on
  Artificial Intelligence}.

\bibitem[{Chiu et~al.(2001)Chiu, Davida, and Litow}]{Chiu2001DivisionIL}
Andrew Chiu, George~I. Davida, and Bruce~E. Litow. 2001.
\newblock Division in logspace-uniform nc1.
\newblock \emph{RAIRO Theor. Informatics Appl.}, 35:259--275.

\bibitem[{Dehghani et~al.(2019)Dehghani, Gouws, Vinyals, Uszkoreit, and
  Kaiser}]{Dehghani2019UniversalT}
Mostafa Dehghani, Stephan Gouws, Oriol Vinyals, Jakob Uszkoreit, and Lukasz
  Kaiser. 2019.
\newblock Universal transformers.
\newblock In \emph{International Conference on Learning Representations}.

\bibitem[{Dettmers et~al.(2022)Dettmers, Lewis, and
  Zettlemoyer}]{dettmers2022gptint}
Tim Dettmers, Mike Lewis, and Luke Zettlemoyer. 2022.
\newblock \href {https://openreview.net/forum?id=dXiGWqBoxaD} {{GPT}3.int8():
  8-bit matrix multiplication for transformers at scale}.
\newblock In \emph{Advances in Neural Information Processing Systems}.

\bibitem[{Devlin et~al.(2019)Devlin, Chang, Lee, and
  Toutanova}]{devlin-etal-2019-bert}
Jacob Devlin, Ming-Wei Chang, Kenton Lee, and Kristina Toutanova. 2019.
\newblock \href {https://doi.org/10.18653/v1/N19-1423} {{BERT}: Pre-training of
  deep bidirectional transformers for language understanding}.
\newblock In \emph{Proceedings of the 2019 Conference of the North {A}merican
  Chapter of the Association for Computational Linguistics: Human Language
  Technologies}.

\bibitem[{Elhage et~al.(2021)Elhage, Nanda, Olsson, Henighan, Joseph, Mann,
  Askell, Bai, Chen, Conerly, DasSarma, Drain, Ganguli, Hatfield-Dodds,
  Hernandez, Jones, Kernion, Lovitt, Ndousse, Amodei, Brown, Clark, Kaplan,
  McCandlish, and Olah}]{elhage2021mathematical}
Nelson Elhage, Neel Nanda, Catherine Olsson, Tom Henighan, Nicholas Joseph, Ben
  Mann, Amanda Askell, Yuntao Bai, Anna Chen, Tom Conerly, Nova DasSarma, Dawn
  Drain, Deep Ganguli, Zac Hatfield-Dodds, Danny Hernandez, Andy Jones, Jackson
  Kernion, Liane Lovitt, Kamal Ndousse, Dario Amodei, Tom Brown, Jack Clark,
  Jared Kaplan, Sam McCandlish, and Chris Olah. 2021.
\newblock \href {https://transformer-circuits.pub/2021/framework/index.html} {A
  mathematical framework for transformer circuits}.
\newblock \emph{Transformer Circuits Thread}.

\bibitem[{Finlayson et~al.(2022)Finlayson, Richardson, Sabharwal, and
  Clark}]{finlayson2022instruction}
Matthew Finlayson, Kyle Richardson, Ashish Sabharwal, and Peter Clark. 2022.
\newblock \href {https://aclanthology.org/2022.emnlp-main.27/} {What makes
  instruction learning hard? {A}n investigation and a new challenge in a
  synthetic environment}.
\newblock In \emph{Proceedings of the 2022 Conference on Empirical Methods in
  Natural Language Processing}.

\bibitem[{Greenlaw et~al.(1991)Greenlaw, Hoover, and
  Ruzzo}]{p-complete-problems}
Raymond Greenlaw, James~M. Hoover, and Walter~L. Ruzzo. 1991.
\newblock \href
  {https://era.library.ualberta.ca/items/403292c5-460b-49e6-8b05-9a5a7b45b0d6}
  {A compendium of problems complete for {P}}.
\newblock Technical Report TR91-11, University of Alberta.

\bibitem[{Hahn(2020)}]{hahn-2020-theoretical}
Michael Hahn. 2020.
\newblock \href {https://www.aclweb.org/anthology/2020.tacl-1.11} {Theoretical
  limitations of self-attention in neural sequence models}.
\newblock \emph{Transactions of the Association for Computational Linguistics},
  8:156--171.

\bibitem[{Hao et~al.(2022)Hao, Angluin, and Frank}]{hao2022}
Yiding Hao, Dana Angluin, and Robert Frank. 2022.
\newblock \href {https://doi.org/10.1162/tacl_a_00490} {Formal language
  recognition by hard attention transformers: Perspectives from circuit
  complexity}.
\newblock \emph{Transactions of the Association for Computational Linguistics},
  10:800--810.

\bibitem[{Hesse(2001)}]{hesse2001division}
William Hesse. 2001.
\newblock Division is in uniform {$TC^0$}.
\newblock In \emph{International Colloquium on Automata, Languages, and
  Programming}, pages 104--114.

\bibitem[{Immerman(2012)}]{immerman2012descriptive}
Neil Immerman. 2012.
\newblock \emph{Descriptive complexity}.
\newblock Springer Science \& Business Media.

\bibitem[{Jones and Laaser(1976)}]{JONES1976105}
Neil~D. Jones and William~T. Laaser. 1976.
\newblock \href {https://doi.org/https://doi.org/10.1016/0304-3975(76)90068-2}
  {Complete problems for deterministic polynomial time}.
\newblock \emph{Theoretical Computer Science}, 3(1):105--117.

\bibitem[{Ladner(1975)}]{ladner1975circuit}
Richard~E Ladner. 1975.
\newblock The circuit value problem is log space complete for {P}.
\newblock \emph{ACM SIGACT News}, 7(1):18--20.

\bibitem[{Lewis and Papadimitriou(1982)}]{Lewis1982SymmetricSC}
Harry~R. Lewis and Christos~H. Papadimitriou. 1982.
\newblock Symmetric space-bounded computation.
\newblock \emph{Theoretical Computer Science}, 19:161--187.

\bibitem[{Merrill et~al.(2021)Merrill, Ramanujan, Goldberg, Schwartz, and
  Smith}]{merrill2020parameter}
William Merrill, Vivek Ramanujan, Yoav Goldberg, Roy Schwartz, and Noah~A.
  Smith. 2021.
\newblock \href {https://doi.org/10.18653/v1/2021.emnlp-main.133} {Effects of
  parameter norm growth during transformer training: Inductive bias from
  gradient descent}.
\newblock In \emph{Proceedings of the 2021 Conference on Empirical Methods in
  Natural Language Processing}.

\bibitem[{Merrill(2020)}]{merrill2021counter}
William~Cooper Merrill. 2020.
\newblock \href {https://arxiv.org/abs/2004.06866} {On the linguistic capacity
  of real-time counter automata}.
\newblock \emph{ArXiv}, abs/2004.06866.

\bibitem[{Merrill et~al.(2022)Merrill, Sabharwal, and
  Smith}]{merrill2022SatAttnTC0}
William~Cooper Merrill, Ashish Sabharwal, and Noah~A. Smith. 2022.
\newblock \href {https://aclanthology.org/2022.tacl-1.49/} {Saturated
  transformers are constant-depth threshold circuits}.
\newblock \emph{Transactions of the Association for Computational Linguistics},
  10.

\bibitem[{Pérez et~al.(2019)Pérez, Marinković, and Barceló}]{perez2019on}
Jorge Pérez, Javier Marinković, and Pablo Barceló. 2019.
\newblock \href {https://openreview.net/forum?id=HyGBdo0qFm} {On the {Turing}
  completeness of modern neural network architectures}.
\newblock In \emph{International Conference on Learning Representations}.

\bibitem[{Raffel et~al.(2020)Raffel, Shazeer, Roberts, Lee, Narang, Matena,
  Zhou, Li, and Liu}]{t5}
Colin Raffel, Noam~M. Shazeer, Adam Roberts, Katherine Lee, Sharan Narang,
  Michael Matena, Yanqi Zhou, Wei Li, and Peter~J. Liu. 2020.
\newblock \href {http://jmlr.org/papers/v21/20-074.html} {Exploring the limits
  of transfer learning with a unified text-to-text transformer}.
\newblock \emph{Journal of Machine Learning Research}, 21(140).

\bibitem[{Reingold(2008)}]{Reingold2008UndirectedCI}
Omer Reingold. 2008.
\newblock Undirected connectivity in log-space.
\newblock \emph{Journal of the ACM}, 55:17:1--17:24.

\bibitem[{Valiant(1979)}]{Valiant1979TheCO}
Leslie~G. Valiant. 1979.
\newblock The complexity of computing the permanent.
\newblock \emph{Theoretical Computer Science}, 8:189--201.

\bibitem[{Vaswani et~al.(2017)Vaswani, Shazeer, Parmar, Uszkoreit, Jones,
  Gomez, Kaiser, and Polosukhin}]{vaswani2017attention}
Ashish Vaswani, Noam Shazeer, Niki Parmar, Jakob Uszkoreit, Llion Jones,
  Aidan~N Gomez, {\L}ukasz Kaiser, and Illia Polosukhin. 2017.
\newblock \href
  {https://proceedings.neurips.cc/paper/2017/file/3f5ee243547dee91fbd053c1c4a845aa-Paper.pdf}
  {Attention is all you need}.
\newblock In \emph{Advances in Neural Information Processing Systems},
  volume~30. Curran Associates, Inc.

\end{thebibliography}
\bibliographystyle{acl_natbib}

\appendix

\section{Iterated $p$-Precision Float Addition} \label{sec:addition}

We interpret a $p$-bit string $x$ as a $p$-precision float by taking the first $p/2$ bits\footnote{We assume w.l.o.g.\ that $p$ is even.} of $x$ as a signed integer $m$ encoding the \emph{mantissa} and the remaining $p/2$ bits of $x$ as another signed integer $e$ encoding the \emph{exponent}. 
A float with mantissa $m$ and exponent $e$,
denoted $\langle m, e \rangle$,
encodes $m \cdot 2^e$.

Computing the sum of $n$ $n$-bit integers (known as iterated addition or simply summation) is well-known to be in uniform $\TC^0$~\cite{hesse2001division,Chiu2001DivisionIL}. We leverage this fact to show that the same holds for the sum of $n$ $\O(\log n)$-precision floats.
A subtlety of adding $p$-precision floats is that their sum can require more than $p$ bits to represent precisely as a float. For instance, while each of $2^r$ and $1$ is representable with a (signed) mantissa of only $2$ bits, their exact sum, $2^r + 1$, requires a mantissa of $r+1$ bits. Hence, $p$-precision transformers must sacrifice some precision when performing summation.

\newcommand{\Imin}{I^{\min}}
\newcommand{\Imax}{I^{\max}}

We define float addition by mapping the floats to integers, adding the integers exactly, and then mapping the sum back to a float (with possible loss of precision).
Let $\Imax_q = 2^q - 1$ be the greatest $q$-bit signed integer, and $\Imin_q = - \Imax_q$. \newB{Let $F_p^{\max}$ be the greatest value representable by a $p$-precision float.}
Since the exponent of a float $\phi$ can be negative \newB{and represent a fraction, we rescale $\phi$ by $2^{-\Imin_{p/2}}$ when mapping} it to an integer $g_p(\phi)$:


\begin{definition} \label{def:float-to-int}
    The integer mapping of a $p$-bit float $\phi = \langle m, e \rangle$ is defined as $g_p(\phi) = m \cdot 2^{e - \Imin_{p/2}}$.
\end{definition}

\begin{definition} \label{def:int-to-float}
    The $p$-truncated float mapping of an integer $z$ is defined as $f_p(z) = \langle m, e \rangle$
    where\footnote{For $x \neq 0$, $\sizeof(x) = \floor{\log \abs{x}} + 2$; $\sizeof(0) = 2$. For $y \geq 0$, $\rshift(x, y)$ right-shifts $x$ by $y$ bits}
    \begin{align*}
        m &= \rshift(z, \max \{0, \sizeof(z) - p/2 \} ) \\
        e &= \sizeof(z) - \sizeof(m) + I_{p/2}^{\min}
    \end{align*}
\newB{when $e \leq \Imax_{p/2}$;} otherwise (i.e., when $z > F_p^{\max}$), we set $m = e = I_{p/2}^{\max}$ to properly handle overflow.
\end{definition}

\begin{definition}[Iterated $p$-precision float addition]
We define the sum of $k$ $p$-precision floats as
\begin{equation*}
    \bigoplus_{i=1}^k \phi_i = f_p \left( \sum_{i=1}^k g_p(\phi_i) \right) .
\end{equation*}
\end{definition}

We first verify that \autoref{def:int-to-float} closely approximates exact addition.

\begin{lemma}
    Let $\phi = \langle e, m \rangle$ be a float such that $\abs{\phi} \leq F_p^{\max}$ and $e \geq I_{p/2}^{\min}$. Then $\phi$ and $f_p(g_p(\phi))$ differ by a factor of at most $1 \pm 2^{-p/2 + 2}$.
\end{lemma}
\begin{proof}
    Let $z = g_p(\phi)$, which is well-defined because of the precondition $e \geq \Imin_{p/2}$ of the lemma. Let $\phi' = \langle m', e' \rangle = f_p(z)$.
    
    First consider the easy case where $\sizeof(z) \leq p/2$. Then $m' = z$ and $e' = \Imin_{p/2}$ from \autoref{def:int-to-float}. Since $z = m \cdot 2^{e - \Imin_{p/2}}$ by \autoref{def:float-to-int}, it follows that $\phi$ and $\phi'$ have exactly the same value.

    Now assume $\sizeof(z) > p/2$. It follows from the precondition $\abs{\phi} \leq F_p^{\max}$ of the lemma that there is no overflow when applying \autoref{def:int-to-float} to compute $\langle m', e' \rangle$. Thus $m'$ consists of the $p/2$ highest-order bits (including the sign bit) of $z$ and $e' = \ell + \Imin_{p/2}$, where $\ell = \sizeof(z) - p/2$ is the number of bits truncated from $z$ to obtain $m'$. Let $\delta$ denote the (non-negative) integer formed by the $\ell$ lowest-order bits of $z$ that are truncated. Then $\delta \leq 2^{\ell} - 1 = 2^{\sizeof(z) - p/2} - 1 < z \cdot 2^{-p/2 + 2}$.

    Recall that the value of $\phi$ is $g_p(\phi) \cdot 2^{-\Imin_{p/2}} = z \cdot 2^{-\Imin_{p/2}}$. By the above argument, we also have that the value of $\phi'$ is within $(z \pm \delta) \cdot 2^{-\Imin_{p/2}}$, which is within $z \cdot (1 \pm 2^{-p/2 + 2}) \cdot 2^{-\Imin_{p/2}}$. Thus, $\phi$ and $\phi'$ are within a factor of $1 \pm 2^{-p/2 + 2}$ of each other.
\end{proof}

Finally, we show that, with log precision, computing $\oplus$ (\autoref{def:int-to-float}) is in uniform $\TC^0$.

\begin{lemma}
    Let $p \leq c \log n$ and $\phi = \bigoplus_{i=1}^k \phi_i$, where $k \leq n$ and each $\phi_i$ is $p$-precision. Then $\phi$ is computable by a constant-depth uniform threshold circuit of size $\poly(n)$.
\end{lemma}

\begin{proof}
    Let $N = c \log n + 2 n^c$. We first use \autoref{lem:hao-extended} to map each $\phi_i = \langle m_i, e_i \rangle$ to the integer $z_i = m_i \cdot 2^{e_i - \Imin_{p/2}}$, which has size $\sizeof(m_i) + (e_i - \Imin) \leq p/2 + 2 \cdot 2^{p/2} \leq c \log n + 2 n^c = N$. For $1 \leq i \leq k$, we pad $z_i$ to $N$ bits, and for $k < i \leq N$, we create an $N$-bit integer $z_i = 0$. We can then compute $z = \sum_{i=1}^k z_i$ with a constant-depth uniform threshold circuit of size $\poly(N)$ using the classical construction to sum $N$ $N$-bit integers \citep[cf.][exercise~5.29]{immerman2012descriptive}. The size of this circuit is also polynomial in $n$ by the definition of $N$. Finally, we compute $f^\dagger(z)$ using a constant-depth AND/OR circuit.
\end{proof}

\end{document}